\documentclass[%
reprint,
superscriptaddress,
nofootinbib,
amsmath,amssymb,
aps,longbibliography
]{revtex4-1}

\newif\ifarxiv 
\arxivtrue

\usepackage{amsthm,bm,braket,xparse,xcolor,tikz}
\usepgfmodule{decorations}
\usetikzlibrary{decorations,decorations.pathmorphing,patterns}

\AtBeginDocument{%
    \newwrite\bibnotes
    \def\bibnotesext{Notes.bib}
    \immediate\openout\bibnotes=\jobname\bibnotesext
    \immediate\write\bibnotes{@CONTROL{REVTEX41Control}}
    \immediate\write\bibnotes{@CONTROL{%
    apsrev41Control,author="08",editor="1",pages="1",title="0",year="0"}}
     \if@filesw
     \immediate\write\@auxout{\string\citation{apsrev41Control}}%
    \fi
  }%
  
\usepackage{hyperref}
\usepackage{xcolor}
\definecolor{mylinkcolor}{rgb}{0,0,0.8} 
\hypersetup{unicode=true, %
  bookmarksnumbered=false,bookmarksopen=false,bookmarksopenlevel=1, %
  breaklinks=true,pdfborder={0 0 0},colorlinks=true}%
\hypersetup{%
  anchorcolor=mylinkcolor,citecolor=mylinkcolor, %
  filecolor=mylinkcolor,linkcolor=mylinkcolor, %
  menucolor=mylinkcolor,runcolor=mylinkcolor, %
  urlcolor=mylinkcolor}%

\newtheorem{theorem}{Theorem}
\newtheorem{lemma}{Lemma}
\newtheorem{corollary}{Corollary}

\theoremstyle{definition}

\newtheorem{remark}{Remark}

\newcommand{\bbr}{\mathbb{R}}

\newcommand{\bbn}{\mathbb{N}}

\let\inner\relax
\NewDocumentCommand\inner{mg}{%
	\ensuremath{\left\langle #1, \! \IfNoValueTF{#2}{#1}{#2}\right\rangle}%
}
\let\outer\relax
\NewDocumentCommand\outer{mg}{%
	\ensuremath{\ket{#1}\!\! \IfNoValueTF{#2}{\bra{#1}}{\bra{#2}}}%
}

\newcommand{\tr}[1]{\operatorname{Tr}\left[#1\right]}

\newcommand{\bobk}[1]{Bob$^{(#1)}$}
\newcommand{\Yk}[1]{Y^{(#1)}}
\newcommand{\cE}{\mathcal{E}}
\newcommand{\cH}{\mathcal{H}}
\newcommand{\cS}{\mathcal{S}}
\newcommand{\Bk}[1]{B^{(#1)}}
\newcommand{\yk}[1]{y_{#1}}
\newcommand{\bk}[1]{b_{#1}}
\newcommand{\eff}[2]{F_{#1|#2}}
\newcommand{\chsh}{I_{\mathrm{CHSH}}}
\newcommand{\chshexp}[1]{\chsh^{(#1)}}

\newcommand{\dup}{d^{\uparrow}}
\newcommand{\ddown}{d^{\downarrow}}

\newcommand{\id}{\mathbb{I}}

\renewcommand{\epsilon}{\varepsilon}

\begin{document}
\title{Arbitrarily many independent observers can share the nonlocality of a single maximally entangled qubit pair}
\author{Peter~J.~Brown}
\email{peter.brown@ens-lyon.fr}
\affiliation{ENS Lyon, LIP, F-69342, Lyon Cedex
	07, France}
\affiliation{Department of Mathematics, University of York, Heslington, York, YO10 5DD, United Kingdom}
\author{Roger~Colbeck}
\email{roger.colbeck@york.ac.uk}
\affiliation{Department of Mathematics, University of York, Heslington, York, YO10 5DD, United Kingdom}

\date{$26^{\text{th}}$ August 2020}

\begin{abstract}
  Alice and Bob each have half of a pair of entangled qubits. Bob measures his half and then passes his qubit to a second Bob who measures again and so on. The goal is to maximize the number of Bobs that can have an expected violation of the Clauser-Horne-Shimony-Holt (CHSH) Bell inequality with the single Alice. This scenario was introduced in [Phys.\ Rev.\ Lett.\ {\bf 114}, 250401 (2015)] where the authors mentioned evidence that when the Bobs act independently and with unbiased inputs then at most two of them can expect to violate the CHSH inequality with Alice. Here we show that, contrary to this evidence, arbitrarily many independent Bobs can have an expected CHSH violation with the single Alice. Our proof is constructive and our measurement strategies can be generalized to work with a larger class of two-qubit states that includes all pure entangled two-qubit states. Since violation of a Bell inequality is necessary for device-independent tasks, our work represents a step towards an eventual understanding of the limitations on how much device-independent randomness can be robustly generated from a single pair of qubits.
\end{abstract}
\maketitle

\ifarxiv\section{Introduction}\else\noindent{\it Introduction.|}\fi
Quantum theory predicts the possibility of measuring correlations that cannot be explained by standard notions of causality~\cite{Bell}. These `non-local' correlations are a crucial resource for device-independent tasks such as key distribution~\cite{MayersYao,BHK,ABGMPS}, randomness expansion~\cite{ColbeckThesis,CK2,PAMBMMOHLMM} and randomness amplification~\cite{CR_free}. The idea is that because such correlations cannot be explained using local hidden variables, the individual outcomes are random and unpredictable to an adversary~\cite{Ekert,MAG,BHK}. Several recent advances have enabled the first experimental demonstrations of device-independence~\cite{diexperiment3,diexperiment2,diexperiment1}. However, further theoretical and experimental advances are needed before this can become a practical technology.

In this work we study fundamental limits of non-locality, asking whether a single pair of entangled qubits could generate a long sequence of non-local correlations. Such an approach could be useful for situations where a significant bottleneck lies in the state generation, such as in nitrogen vacancy based experiments~\cite{Hensen&}. We study a scenario in which a single Alice tries to establish non-local correlations with a sequence of Bobs who sequentially measure one half of an entangled qubit pair. An additional restriction we impose is that each Bob in the sequence can only send a single qubit (his post-measurement state) to the next. In particular, the classical information pertaining to measurement choices and outcomes of each Bob is not shared. It is in this sense that the Bobs act independently of one another.

This sequential scenario (see Fig.~\ref{fig:scenario}\ifarxiv\ and Sec.~\ref{sec:scenario}\fi) was introduced in~\cite{SGGP}. There it was shown that by modifying the input distributions of the Bobs, so that one of the inputs is highly favoured, an unbounded number of Bobs could each have an expected CHSH violation~\cite{CHSH} with the single Alice who measures once. However, the authors also mentioned numerical evidence suggesting that if the input distributions were not modified (each Bob chooses a binary input uniformly at random) then at most two Bobs would be able to have an expected CHSH violation with Alice.

By constructing an explicit measurement strategy, we show that, contrary to what was previously thought, there is no bound on the number of independent Bobs (with uniform inputs) that can have an expected violation of the CHSH inequality with Alice. We exhibit a class of initial two-qubit entangled states that are capable of achieving an unbounded number of violations, which includes all pure two-qubit entangled states. A previous work~\cite{MMH} claimed a proof that at most two unbiased Bobs could achieve an expected CHSH violation with a single Alice, in line with the numerical evidence mentioned in~\cite{SGGP}. However, within~\cite{MMH} there is an implicit assumption that the sharpness of the two measurements that each Bob used is equal. 
Equal sharpness was also taken in~\cite{SDHSGB} to show that at most two Bobs can violate \emph{any} 2-outcome Bell inequality with Alice in this scenario when starting with a maximally entangled state. 
In the present \ifarxiv paper \else Letter \fi we show that these limitations can be overcome by using measurements with unequal sharpness, highlighting the advantage of considering the most general measurement strategies.

\ifarxiv\section{The sequential CHSH scenario}\label{sec:scenario}\else\bigskip\noindent{\it The sequential CHSH scenario.|}\fi
We consider the following measurement scenario wherein a single party (Alice) attempts to share nonlocal correlations with $n$ independent parties (\bobk{1}, \dots, \bobk{n}) using only a single maximally entangled qubit pair, one half of which is passed between the $n$ parties (see Fig.~\ref{fig:scenario}). We denote the binary input and output of Alice by $X$ and $A$ respectively. For each $k \in \bbn$ we denote the binary input and output of \bobk{k} by $\Yk{k}$ and $\Bk{k}$ respectively. 

\begin{figure}[t]
	\begin{tikzpicture}
	\fill[black!20, rounded corners] (0,2.75) rectangle (1.25,4.0);
	\node[] (alice) at (0.625,3.375) {Alice};
	\node[] (alice_input) at (-0.1,4.5) {$X$};
	\node[] (alice_output) at (1.35,4.5) {$A$};
	\draw[->] (alice_input) -- (alice);
	\draw[->] (alice) -- (alice_output); 
	
	\fill[black!20, rounded corners] (0,0) rectangle (1.25,1.25);
	\node[] (bob1) at (0.625,0.625) {\bobk{1}};
	\node[] (bob1_input) at (-0.1, -0.5) {$Y^{(1)}$};
	\node[] (bob1_output) at (1.35, -0.5) {$B^{(1)}$};
	\draw[->] (bob1_input) -- (bob1);
	\draw[->] (bob1) -- (bob1_output);
	
	\fill[pattern=north east lines, pattern color = blue] (1.8,-0.8) rectangle ++(0.1,1.2);
	\fill[pattern=north east lines, pattern color = blue] (1.8,0.8) rectangle ++(0.1,1.2);
	
	\fill[black!20, rounded corners] (2.5,0) rectangle (3.75,1.25);
	\node[] (bob2) at (3.125,0.625) {\bobk{2}};
	\node[] (bob2_input) at (2.4, -0.5) {$Y^{(2)}$};
	\node[] (bob2_output) at (3.85, -0.5) {$B^{(2)}$};
	\draw[->] (bob2_input) -- (bob2);
	\draw[->] (bob2) -- (bob2_output);
	
	\fill[pattern=north east lines, pattern color = blue] (4.3,-0.8) rectangle ++(0.1,1.2);
	\fill[pattern=north east lines, pattern color = blue] (4.3,0.8) rectangle ++(0.1,1.2);
	
	\node[] (dots) at (4.875,0.625) {\dots};
	\node[] at (4.875,-0.5) {\dots};
	
	\fill[pattern=north east lines, pattern color = blue] (5.3,-0.8) rectangle ++(0.1,1.2);
	\fill[pattern=north east lines, pattern color = blue] (5.3,0.8) rectangle ++(0.1,1.2);
	
	\fill[black!20, rounded corners] (6,0) rectangle (7.25,1.25);
	\node[] (bobn) at (6.625,0.625) {\bobk{n}};
	\node[] (bobn_input) at (5.9, -0.5) {$Y^{(n)}$};
	\node[] (bobn_output) at (7.35, -0.5) {$B^{(n)}$};
	\draw[->] (bobn_input) -- (bobn);
	\draw[->] (bobn) -- (bobn_output);
	
	\node[] (state) at (-0.8,2.0) {$\rho_{A\Bk{1}}$};
	\draw[->, decorate, decoration={snake,amplitude=.4mm,segment length=2mm,post length=1mm}] (state) -- (alice);
	\draw[->, decorate, decoration={snake,amplitude=.4mm,segment length=2mm,post length=1mm}] (state) -- (bob1);
	
	\draw[->, decorate, decoration={snake,amplitude=.4mm,segment length=2mm,post length=1mm}] (bob1) -- (bob2);
	\draw[->, decorate, decoration={snake,amplitude=.4mm,segment length=2mm,post length=1mm}] (bob2) -- (dots);
	\draw[->, decorate, decoration={snake,amplitude=.4mm,segment length=2mm,post length=1mm}] (dots) -- (bobn);
	\end{tikzpicture}
	\caption{A schematic of the considered sequential CHSH scenario. All random variables $X,A,Y^{(i)},B^{(i)}$ for $i = 1,\dots,n$ have only two outcomes. A quantum state $\rho_{A\Bk{1}}$ is initially shared between Alice and \bobk{1}. After \bobk{1} has performed his randomly selected measurement and recorded the outcome he passes the qubit post-measurement state to \bobk{2} who repeats this process. Only the qubit post-measurement states are sent to the next Bob (the classical information regarding the measurement inputs and outputs are not conveyed). Each Bob also knows his position in the sequence.}
	\label{fig:scenario}
\end{figure}
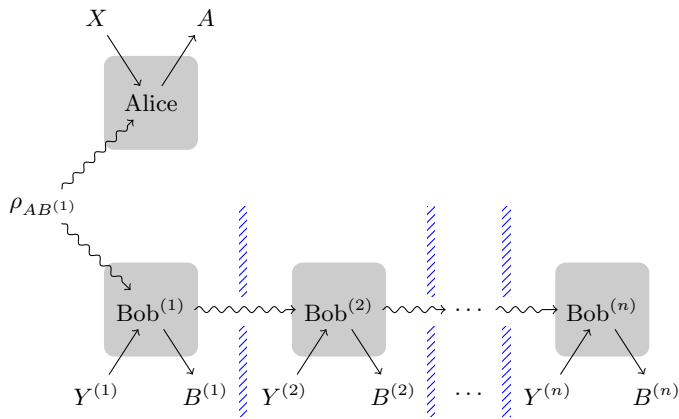

To begin, a two-qubit state $\rho_{A\Bk{1}}$ is shared between Alice and \bobk{1}. \bobk{1} then proceeds by choosing a uniformly random input, performing the corresponding measurement and recording the outcome. The post-measurement state is then sent to \bobk{2}. Suppose \bobk{1} performed the measurement according to $\Yk{1}=y$ and received the outcome $\Bk{1}=b$. The post-measurement state can be described by the L\"uders rule\footnote{Using the L\"uders rule is optimal in terms of the information retained in the post-measurement state|see \ifarxiv Appendix~\ref{app:opt}\else the Supplemental Material\fi.}
\begin{equation}
\rho_{A\Bk{1}} \mapsto \left(\id \otimes \sqrt{\eff{b}{y}^{(1)}} \right) \rho_{A\Bk{1}} \left(\id \otimes \sqrt{\eff{b}{y}^{(1)}} \right),
\end{equation}
where $\eff{b}{y}^{(1)}$ is the POVM effect corresponding to outcome $b$ of \bobk{1}'s measurement for input $y$. As each Bob is assumed to act independently of all previous Bobs, \bobk{2} is ignorant of the values of $\Yk{1}$ and $\Bk{1}$. Thus, the state shared between Alice and \bobk{2} is averaged over the inputs and outputs of \bobk{1}, i.e.
\begin{equation}\label{eq:update_rule}
\rho_{A\Bk{2}} = \frac12 \sum_{\bk{1},\yk{1}} \left(\id \otimes \sqrt{\eff{\bk{1}}{\yk{1}}^{(1)}} \right) \!\rho_{A\Bk{1}}\! \left(\id \otimes \sqrt{\eff{\bk{1}}{\yk{1}}^{(1)}} \right).
\end{equation} 
Repeating this process, we can compute the state shared between Alice and \bobk{k}.

We are concerned with the expected CHSH value between Alice and each of the Bobs independently. Given a conditional probability distribution $p_{AB^{(k)}|XY^{(k)}}$ on the binary random variables $A$, $B^{(k)}$, $X$ and $Y^{(k)}$, the \emph{CHSH value} is defined as 
\begin{equation}
\begin{aligned}
\chshexp{k} :=&\, 2[\,p(A=B^{(k)}|00) + p(A=B^{(k)}|01)  \\
 +&\, p(A=B^{(k)}|10) + p(A\neq B^{(k)}|11) - 2].
\end{aligned}
\end{equation}
When Alice and \bobk{k} behave classically (with shared randomness) or do not share entanglement then ${\chshexp{k}\leq2}$, which is the CHSH inequality.

The aim of this work is to explain how, given $n\in \bbn$, we can define a sequence of pairs of POVMs for \bobk{1}, \dots, \bobk{n} such that $\chshexp{k} > 2$ for all $k = 1,\dots,n$. That is, by initially sharing two qubits in an appropriate entangled state, there is no bound on the number of independent Bobs that can expect to violate the CHSH inequality with a single Alice.

\ifarxiv\section{The measurement strategy}\else\bigskip\noindent{\it The measurement strategy.|}\fi
This uses two-outcome POVMs, $\{E,\id-E\}$, where $E$ has the form $E=\tfrac12\left(\id+\gamma\sigma_{\bm{r}}\right)$. Here $\bm{r}\in\bbr^3$ with $\|\bm r\|=1$, $\gamma \in [0,1]$ is the \emph{sharpness} of the measurement, and $\sigma_{\bm{r}}=r_1\sigma_1+r_2\sigma_2+r_3\sigma_3$, where $\{\sigma_1,\sigma_2,\sigma_3\}$ are the Pauli operators. Since there are only two outcomes, a single effect is enough to define the POVM. Note that if $\gamma=1$ the measurement is sharp (projective), and the post-measurement state is unentangled, while if $\gamma=0$, the measurement ignores the state and is equivalent to an unbiased coin toss.

Our measurement strategy for nonlocality sharing between Alice and $n$ Bobs is fully specified by $n + 1$ parameters. In this strategy Alice's POVMs are defined by the effects
\begin{align}\label{eq:alice_measurements}
A_{0|0} &:= \tfrac12(\id + \cos(\theta) \sigma_3 + \sin(\theta) \sigma_1)\quad\text{and}\\
A_{0|1} &:= \tfrac12(\id + \cos(\theta) \sigma_3 - \sin(\theta) \sigma_1)
\end{align}
for some $\theta \in(0, \pi/4]$. For each $k = 1, \dots, n$, \bobk{k}'s POVMs are defined by the effects
\begin{align}\label{eq:Bob_measurements}
B_{0|0}^{(k)} &:= \tfrac12(\id+ \sigma_3)\quad\text{and}\\
B_{0|1}^{(k)} &:= \tfrac12(\id +\gamma_k \sigma_1).
\end{align}
For these measurements and with an initial state $\rho_{A\Bk{1}} = \outer{\Psi}$ with $\ket{\Psi} = \tfrac{1}{\sqrt{2}}(\ket{00} + \ket{11})$, the expected CHSH value between Alice and \bobk{k} is given by
\begin{equation}\label{eq:chsh}
\chshexp{k} = 2^{2-k} \left(\gamma_k \sin(\theta) + \cos(\theta) \prod_{j=1}^{k-1} \left(1+ \sqrt{1-\gamma_j^2}\right)\right).
\end{equation}
(see \ifarxiv Appendix~\ref{app:CHSH_deriv}\else the Supplemental Material\fi). 

\ifarxiv\section{An unbounded number of violations}\label{sec:unbound}\else\bigskip\noindent{\it An unbounded number of violations.|}\fi
Our aim is to choose the values of $\gamma_k$ and $\theta$ to achieve $\chshexp{k}>2$ for arbitrarily many values of $k$. Using~\eqref{eq:chsh}, we require
\begin{equation}\label{eq:violation_constraint}
\gamma_k > \frac{2^{k-1} - \cos(\theta) \prod_{j=1}^{k-1} \left(1 + \sqrt{1-\gamma_j^2}\right)}{\sin(\theta)}.
\end{equation}
Our strategy is as follows. Let $\epsilon > 0$, $\gamma_1(\theta)=(1+\epsilon)\tfrac{1-\cos(\theta)}{\sin(\theta)}$ and for $k\in\{2,\ldots,n\}$ recursively set
\begin{equation}\label{eq:gammak_alt}
\!\!\gamma_k(\theta)=\begin{cases}(1+\epsilon)\frac{2^{k-1}-\cos(\theta)P_k}{\sin(\theta)}&\text{if }\gamma_{k-1}(\theta)\in(0,1)\\
    \infty&\text{otherwise}\end{cases}\!,
\end{equation}
where $P_k=\prod_{j=1}^{k-1} \left(1 + \sqrt{1-\gamma_j^2(\theta)}\right)$, and $\gamma_k(\theta)=\infty$ indicates that violation of $\chshexp{k}>2$ is not possible for this $k$.


The following lemma shows that the sharpness parameters of such a measurement strategy (when they are finite) are strictly increasing, i.e., $0 < \gamma_1(\theta) < \gamma_2(\theta) < \dots$. 
\begin{lemma}\label{lem:gammak_increasing}
	Let $\theta \in (0, \pi/4]$ and $\epsilon > 0$. Then the sequence $(\gamma_k(\theta))_{k \in \bbn}$ as defined by~\eqref{eq:gammak_alt} 
	with $\gamma_1(\theta) = (1+\epsilon)\tfrac{1-\cos(\theta)}{\sin(\theta)}$, is positive and increasing. Moreover, the subsequence consisting of all finite terms is a strictly increasing sequence. 
\end{lemma}

In addition to the sequence being monotonically increasing, each term in the sequence also has a vanishing limit as $\theta$ approaches $0$. 
\begin{lemma}\label{lem:gammak_limit}
	Let $\epsilon > 0$ and let $(\gamma_n(\theta))_{n \in \bbn}$ be the sequence defined by~\eqref{eq:gammak_alt} 
		with $\gamma_1(\theta) = (1+\epsilon)\tfrac{1-\cos(\theta)}{\sin(\theta)}$. Then for any $n \in \bbn$ there exists some $\theta_n \in (0,\pi/4]$ such that for all $k \leq n$  and $\theta \in (0, \theta_n)$, \ifarxiv
		\begin{equation}\else $\fi
		\gamma_k(\theta) < 1
		\ifarxiv.\end{equation}\else$.
		
		\fi Moreover, we have \ifarxiv
		\begin{equation}\else $\fi
		\lim_{\theta \rightarrow 0^+}\gamma_n(\theta) = 0
		\ifarxiv\end{equation}\else$ \fi
              for all $n \in \bbn$.
            \end{lemma}

Using the two lemmas above, the proofs of which are given in \ifarxiv{Appendices~\ref{app:lem1} and~\ref{app:lem2}}\else the Supplemental Material\fi, we can now state and prove our main result.
\begin{theorem}\label{thm:main}
	For each $n \in \bbn$, there exists a sequence $(\gamma_k(\theta))_{k=1}^n$ and a $\theta_n \in (0, \pi/4]$ such that $\chshexp{k} > 2$ for all $k = 1,\dots,n$.
	\begin{proof}
	By Lemma~\ref{lem:gammak_limit}, we know that there exists some $\theta_n \in (0, \pi/4]$ such that $\gamma_n(\theta_n) < 1$. By Lemma~\ref{lem:gammak_increasing}, we then also have that $0< \gamma_1(\theta_n) < \gamma_2(\theta_n) < \dots < \gamma_n(\theta_n) < 1$.  As such, this value of $\theta_n$ defines a sequence of $n$ valid sharpness parameters which, by construction, each satisfy the conditions \eqref{eq:violation_constraint} and thus each achieve an expected CHSH value greater than $2$. 
	\end{proof}
\end{theorem}

The above theorem shows that an unbounded number of independent Bobs are able to violate the CHSH inequality with a single Alice, by sequentially measuring one half of a maximally entangled pair. 
\begin{figure}[t]
	\includegraphics[scale=0.45]{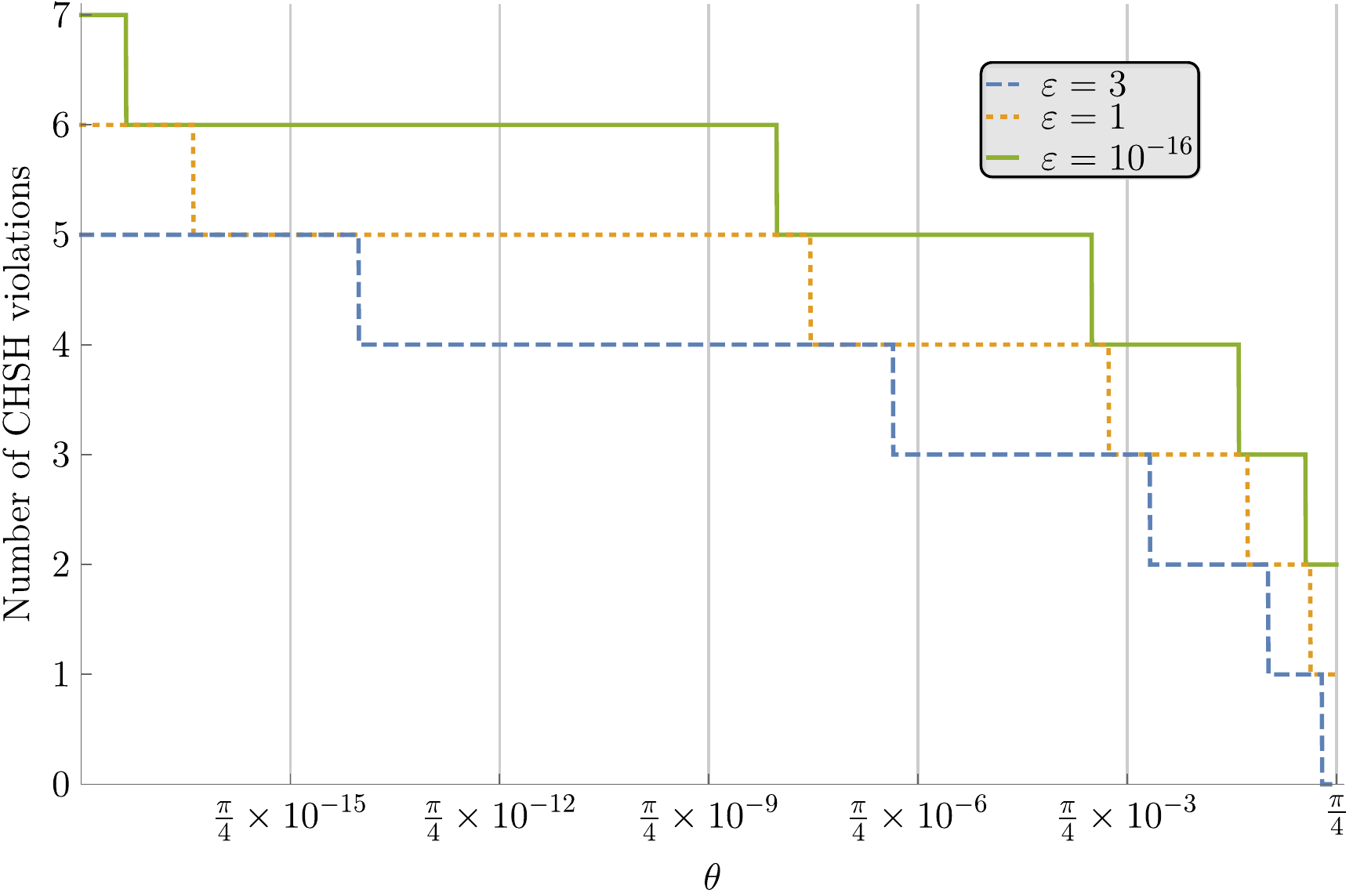}
	\caption{The number of Bobs that can violate CHSH using our strategy as a function of $\theta$. This was computed numerically for several values of $\epsilon$ by finding, for a fixed $\theta$, the maximum $k$ such that $\gamma_k(\theta) <1$. As the $\theta$ axis is plotted with a log-scale, the nature of the plot indicates that the sequence $(\theta_n)$ will likely decrease faster than exponentially in $n$.}
	\label{fig:numerics}
\end{figure}

\begin{remark}\label{rem:theta_n}
	To increase the number $n$ of Bobs violating the CHSH inequality, we require smaller values of $\theta_n \in (0, \pi/4]$. By upper and lower bounding $\gamma_n$ by polynomials of $\theta$, we find evidence to suggest that $\theta_n$ must decrease double-exponentially fast with $n$. A more detailed discussion of this is given in \ifarxiv Appendix~\ref{app:theta}\else the Supplemental Material\fi. Our findings agree with the numerics presented in Fig.~\ref{fig:numerics} where we plot $\theta$ (log-scale) against the number of CHSH violations possible.
\end{remark}

\begin{remark}
  We also investigate the behaviour of the sequence of violations of the CHSH inequality. For the strategy described in this \ifarxiv paper\else Letter\fi, we can bound $\chshexp{k}< 2^{2-k} \left(\gamma_k(\theta_n) \theta_n + 2^{k-1}\right)<  2 + 2^{2-k} \theta_n$ for $k=1,\dots,n$, where we have used $1+ \sqrt{1-\gamma_j^2(\theta_n)}<2$ for each $j=1,\dots,n$ and $\gamma_k(\theta_n)<1$. Therefore, the sizes of the CHSH violations must decrease as least as fast as $\theta_n$.
\end{remark}

\begin{remark}
	The analysis of this section was restricted to the setting where Alice and \bobk{1} initially share a maximally entangled qubit pair. However, the measurement strategy presented can be readily extended to a larger class of two qubit states. In particular, if Alice and \bobk{1} initially share any pure two-qubit entangled state then it is also possible to define a measurement strategy that gives rise to arbitrarily many CHSH violations in the scenario under consideration. The larger class of states and their respective measurement strategies are discussed in \ifarxiv Appendix~\ref{app:CHSH_deriv}\else the Supplemental Material\fi.
\end{remark}

\ifarxiv\section{Discussion}\else\bigskip\noindent{\it Discussion.|}\fi
In this work we have shown that it is possible for arbitrarily many independent Bobs to violate the CHSH inequality with a single Alice using only a single maximally entangled qubit pair and while making uniformly distributed inputs, answering the problem posed in~\cite{SGGP} in the negative and overturning a claim in~\cite{MMH} due to an implicit assumption made there. In addition, we showed that the same can be achieved if the parties initially share any pure, entangled two-qubit state. We also provided bounds on the size of the violations achievable with our strategy, giving evidence that they decrease double-exponentially fast with the number of Bobs.

The magnitude of the violations suggests that our measurement strategies, whilst able to achieve an unbounded number of violations, are not suited to device-independent tasks like randomness expansion in which the quantity of randomness certifiable increases with the size of the CHSH violation.  However, it is known that by using the tilted CHSH inequality~\cite{AMP2012} and allowing one party to choose from an exponentially large number of measurements, an unbounded number of random bits can (in principle, but not in a robust way) be certified from a single pair of entangled qubits in a device-independent manner~\cite{CJAHWA}. More recently, in~\cite{BBS} it was shown that more than two bits of local randomness can be robustly certified by performing sequential measurements on a two-qubit system and with each party choosing from at most three measurements. Further work is needed to understand the limitations on robust device-independent randomness expansion from a single pair of entangled qubits.

The scenario of the present \ifarxiv paper \else Letter \fi has also been studied for the related tasks of steering~\cite{SDMM,SDHSGB} and entanglement witnessing~\cite{BMSS}.  Since the presence of non-locality implies the possibility of steering and the presence of entanglement, our work shows that it is also possible to achieve both of these with arbitrarily many independent Bobs using only a single pair of entangled qubits. This goes against some of the results in~\cite{BMSS,SDMM,SDHSGB}, and suggests that it is worth rethinking others that study the space of such correlations~\cite{DGSMM} or look at other Bell inequalities~\cite{KumariPan} or more parties~\cite{SDSSM,MDGRM} after removing the assumption that both measurements used by each party have the same sharpness. 

Several interesting questions remain open. Firstly, can we fully characterise the set of two-qubit states that allow for an unbounded number of CHSH violations? Here, we gave a sufficient condition for two-qubit state to achieve unbounded violations. In the standard case of a single Bob the set of states for which a violation of the CHSH inequality is possible has been fully characterised~\cite{HHH95}. It would be interesting to know whether the conditions presented in \ifarxiv Theorem~\ref{thm:main_generalised} \else the Supplemental Material \fi are both necessary and sufficient for unbounded violations.

Second, in the scenario analysed here we only permit a single qubit to be transmitted between subsequent Bobs. We could also consider the setting in which we additionally allow the Bobs to share classical information, i.e., the inputs and outputs of previous Bobs. Such a setting would open up the possibility of the Bobs using an adaptive strategy. In~\cite{sequential_experiment_1} some steps were made in this direction but the authors' strategy required that the inputs and outputs of the Bobs were also sent to Alice before she made her measurement. Whilst our work here implies that it is also possible in the classically-assisted setting to achieve unbounded violations, it would be interesting to know if the classical communication could be used to produce larger violations and more noise-tolerant strategies. 

Finally, in a different direction, it would be interesting to explore the scenario where we also consider multiple Alices. In particular, how many expected pairwise violations could a sequence of Alices and Bobs achieve from a single pair of qubits if they act independently?

\ifarxiv\begin{acknowledgements}\else\bigskip\noindent{\it Acknowledgements.|}\fi
	PB thanks Mirjam Weilenmann for feedback on an earlier version of the manuscript. This work was supported by EPSRC's Quantum Communications Hub (grant numbers EP/M013472/1 and EP/T001011/1), an EPSRC First Grant (grant number EP/P016588/1) and the French National Research Agency via Project No.\ ANR-18-CE47-0011 (ACOM). The majority of this work was carried out whilst PB was at the University of York.
\ifarxiv\end{acknowledgements}\fi

%

\onecolumngrid
\appendix

\section{Optimality of the L\"uders rule}\label{app:opt}
The following lemma shows that use of the L\"uders rule is optimal in terms of the information retained in the post-measurement state.
\begin{lemma}
  Let $\{E_i\}_{i=1}^s$ be a POVM on $\cH$. Any instrument $\{\cE_i\}_{i=1}^s$ that implements this POVM can be constructed by performing the L\"uders measurement followed by a quantum channel that depends on the outcome. 
\end{lemma}
\begin{proof}
  Let $\cS(\cH)$ be the set of density operators on $\cH$. If $\{K_i^j\}_j$ are the Kraus operators representing the channel $\cE_i$, then, when acting on a state $\rho\in\cS(\cH)$, the post-measurement state on obtaining outcome $i$ is proportional to $\cE_i(\rho)=\sum_jK_i^j\rho(K_i^j)^\dagger$.

  Using the L\"uders update, the post measurement state on receiving outcome $i$ is proportional to $\sqrt{E_i}\rho\sqrt{E_i}$.  We now consider the channel with Kraus operators $\{J_i^j\}_j\cup\{\id-\Pi_{E_i}\}$, where $\Pi_{E_i}$ is the projector onto the support of $E_i$, $J_i^j=K_i^jE_i^{-1/2}$ and the inverse is interpreted as the pseudo-inverse (so that $E_i^{-1/2}E_iE_i^{-1/2}=\Pi_{E_i}$).  These form a valid quantum channel since
  $$\sum_j(J^j_i)^\dagger J^j_i+\id-\Pi_{E_i}=\sum_jE_i^{-1/2}(K_i^j)^\dagger K_i^jE_i^{-1/2}+\id-\Pi_{E_i}=E_i^{-1/2}E_iE_i^{-1/2}+\id-\Pi_{E_i}=\id\,,$$
where we have used $\sum_j(K_i^j)^\dagger K_i^j=E_i$ which follows because $\cE_i$ has corresponding POVM element $E_i$. Furthermore, the concatenation of the L\"uders measurement with this channel corresponds to the map
$$\rho\mapsto(\id-\Pi_{E_i})\sqrt{E_i}\rho\sqrt{E_i}(\id-\Pi_{E_i})+\sum_jJ^j_i\sqrt{E_i}\rho\sqrt{E_i}(J^j_i)^\dagger=\sum_jK_i^j\Pi_{E_i}\rho\Pi_{E_i}(K_i^j)^\dagger=\sum_jK_i^j\rho(K_i^j)^\dagger=\cE_i(\rho)\,,$$
where we have used $K_i^j\Pi_{E_i}=K_i^j$ for all $j$, which follows from Corollary~\ref{cor:KPi} noting that $\Pi_{E_i}\geq\Pi_{(K_i^j)^\dagger K_i^j}$.
In other words, using the L\"uders rule followed by the channel with Kraus operators $\{J^j_i\}_j\cup\{\id-\Pi_{E_i}\}$ is equivalent to the relevant instrument channel.
\end{proof}

\begin{lemma}\label{lem:KPi}
Let $K$ be a linear map from $\cH_A$ to $\cH_B$ and $\Pi_{K^\dagger K}$ be the projector onto the support of $K^\dagger K$. Then $K\Pi_{K^\dagger K}=K$.
\end{lemma}
\begin{proof}
  We first suppose $K$ is of the form $\sum_id_i\ket{i}_B\!\bra{i}_A$, where $\{\ket{i}_A\}$ ($\{\ket{i}_B\}$) is an orthonormal basis for $\cH_A$ ($\cH_B$).  We have $K^\dagger K=\sum_i|d_i|^2\ket{i}_A\!\bra{i}_A$ and $\Pi_{K^\dagger K}=\sum_{i:d_i\neq0}\ket{i}_A\!\bra{i}_A$. Thus, $K\Pi_{K^\dagger K}
  =K$.

More generally, let $K=UDV$ be the singular value decomposition of $K$, where $U:\cH_B\to\cH_B$ and $V:\cH_A\to\cH_A$ are unitary and $D:\cH_A\to\cH_B$ is non-zero only on the diagonal.  We have $K^\dagger K=V^\dagger D^\dagger DV$ and hence $\Pi_{K^\dagger K}=V^\dagger\Pi_{D^\dagger D}V$. Thus, $K\Pi_{K^\dagger K}=UDVV^\dagger\Pi_{D^\dagger D}V=UDV=K$, using the first part of the proof.
\end{proof}
\begin{corollary}\label{cor:KPi}
Let $K$ be a linear map from $\cH_A$ to $\cH_B$ and $\Pi$ be a projector on $\cH_A$ such that $\Pi\geq\Pi_{K^\dagger K}$. Then $K\Pi=K$.
\end{corollary}


\section{Derivation of CHSH expression}\label{app:CHSH_deriv}

In this section we will derive the CHSH value for a more general strategy, recovering the expression \ifarxiv in Equation~\eqref{eq:chsh} \else from the main text \fi as a special case.  The notation here and basis for the strategy is inspired by~\cite{HHH95}. Given a state $\rho_{AB}$ of two qubits, we define $T(\rho_{AB})$ to be the $3\times 3$ matrix whose entries are given by $T_{i,j}(\rho_{AB}) = \tr{\rho_{AB}(\sigma_i \otimes \sigma_j)}$. Considering this matrix for the initial state $\rho_{AB^{(1)}}$ shared between Alice and \bobk{1}, we define $\lambda_0$, $\lambda_1$ be the two largest eigenvalues of $T(\rho_{AB^{(1)}})T^{T}(\rho_{AB^{(1)}})$ and let $\bm{c}_0$, $\bm{c}_1$ be the corresponding orthonormal eigenvectors. We also define $\bm{b}_i = \tfrac{T^{T}(\rho_{AB^{(1)}})\bm{c}_i}{\|T^{T}(\rho_{AB^{(1)}})\bm{c}_i\|}$. Note that for two unit-length vectors $\bm{a}, \bm{b} \in \bbr^{3}$ we have $\tr{\rho(\sigma_{\bm{a}} \otimes\sigma_{\bm{b}})} = (\bm{a}, T(\rho) \bm{b})$ where $(\cdot\,, \cdot)$ is the standard Euclidean inner product for the real vector space $\bbr^3$.

Let the measurement strategy of Alice and \bobk{k} be defined by the effects
\begin{equation}\label{eq:general_strategy}
\begin{aligned}
A_{0|0} &= \tfrac12(\id + \cos(\theta) \sigma_{\bm{c}_0} + \sin(\theta)\sigma_{\bm{c}_1}), \\
A_{0|1} &= \tfrac12(\id + \cos(\theta) \sigma_{\bm{c}_0} - \sin(\theta)\sigma_{\bm{c}_1}), \\
B_{0|0}^{(k)} &= \tfrac12(\id + \sigma_{\bm{b}_0}), \\
B_{0|1}^{(k)} &= \tfrac12(\id + \gamma_k \, \sigma_{\bm{b}_1}).
\end{aligned}
\end{equation}
Defining the \emph{expectation operators} $A_x = A_{0|x} - A_{1|x}$ and $B_y = B_{0|y} - B_{1|y}$, the CHSH value of Alice and \bobk{k} can be written 
\begin{equation}\label{eq:chsh_expression2_app}
\begin{aligned}
\chshexp{k} &= \tr{\rho_{AB^{(k)}}((A_0 + A_1)\otimes B_0^{(k)})} + \tr{\rho_{AB^{(k)}}((A_0 - A_1)\otimes B_1^{(k)})}, \\
&= 2 \cos(\theta) (\bm{c}_0, T(\rho_{AB^{(k)}}) \bm{b}_0) + 2 \gamma_k \sin(\theta) (\bm{c}_1, T(\rho_{AB^{(k)}}) \bm{b}_1)\,.
\end{aligned}
\end{equation}
For the case $k=1$ this gives
\begin{equation}
\begin{aligned}
\chshexp{1} &= 2 \cos(\theta) (\bm{c}_0, T(\rho_{AB^{(1)}}) \bm{b}_0) + 2 \gamma_1 \sin(\theta) (\bm{c}_1, T(\rho_{AB^{(1)}}) \bm{b}_1)\,, \\
&= 2 \cos(\theta) \frac{(\bm{c}_0, T(\rho_{AB^{(1)}}) T^{T}(\rho_{AB^{(1)}}) \bm{c}_0)}{\|T^{T}(\rho_{AB^{(1)}}) \bm{c}_0\|} + 2 \gamma_1 \sin(\theta) \frac{(\bm{c}_1, T(\rho_{AB^{(1)}}) T^{T}(\rho_{AB^{(1)}}) \bm{c}_1)}{\|T^{T}(\rho_{AB^{(1)}}) \bm{c}_1\|}\,, \\
&= 2 \cos(\theta) \|T^{T}(\rho_{AB^{(1)}}) \bm{c}_0\| + 2 \gamma_1 \sin(\theta)\|T^{T}(\rho_{AB^{(1)}}) \bm{c}_1\|, \\
&= 2 \cos(\theta) \sqrt{\lambda_0} + 2 \gamma_1 \sin(\theta)\sqrt{\lambda_1}.
\end{aligned}
\end{equation}
To derive the general expression we will show how to relate $(\bm{c}_0, T(\rho_{AB^{(k)}}) \bm{b}_0)$ to $(\bm{c}_0, T(\rho_{AB^{(1)}}) \bm{b}_0)$ and $(\bm{c}_1, T(\rho_{AB^{(k)}}) \bm{b}_1)$ to $(\bm{c}_1, T(\rho_{AB^{(1)}}) \bm{b}_1)$. 

Let $\rho_{AB^{(k-1)}}$ be the state shared between Alice and \bobk{k-1} prior to \bobk{k-1}'s measurement (and where Alice has not measured). Using the L\"uders update rule the state sent to \bobk{k} is
\begin{equation}\label{eq:tmeasurements_update}
\begin{aligned}
\rho_{AB^{(k)}} &= \frac12 \sum_{b,y} \left(\id \otimes \sqrt{B_{b|y}^{(k)}}\right) \rho_{AB^{(k-1)}} \left(\id \otimes \sqrt{B_{b|y}^{(k)}}\right)  \\
 &= \left(\id \otimes \tfrac12(\id + \sigma_{\bm{b}_0})\right) \rho_{AB^{(k-1)}} \left(\id \otimes \tfrac12(\id + \sigma_{\bm{b}_0})\right) + \left(\id \otimes \tfrac12(\id - \sigma_{\bm{b}_0})\right) \rho_{AB^{(k-1)}} \left(\id \otimes \tfrac12(\id - \sigma_{\bm{b}_0})\right) \\
 &\,+ \left(\id \otimes \sqrt{\tfrac12(\id + \gamma_k \, \sigma_{\bm{b}_1})}\right) \rho_{AB^{(k-1)}} \left(\id \otimes \sqrt{\tfrac12(\id + \gamma_k \, \sigma_{\bm{b}_1})}\right) + \left(\id \otimes \sqrt{\tfrac12(\id - \gamma_k \, \sigma_{\bm{b}_1})}\right) \rho_{AB^{(k-1)}} \left(\id \otimes \sqrt{\tfrac12(\id - \gamma_k \, \sigma_{\bm{b}_1})}\right)  \\
 &=\frac{2 + \sqrt{1-\gamma_{k-1}^2}}{4} \rho_{AB^{(k-1)}} + \frac14 (\id \otimes \sigma_{\bm{b}_0}) \rho_{AB^{(k-1)}} (\id \otimes \sigma_{\bm{b}_0}) + \frac{1-\sqrt{1-\gamma_{k-1}^2}}{4} (\id \otimes \sigma_{\bm{b}_1}) \rho_{AB^{(k-1)}} (\id \otimes \sigma_{\bm{b}_1}). 
\end{aligned}
\end{equation}
The final line follows from a direct calculation using the identity
\begin{equation*}
\sqrt{\tfrac12(\id \pm \gamma_k \, \sigma_{\bm{b}_1})} = \frac{(\sqrt{1+\gamma_k} + \sqrt{1-\gamma_k}) \id \pm (\sqrt{1+\gamma_k} - \sqrt{1-\gamma_k}) \sigma_{\bm{b}_1}}{2\sqrt{2}}.
\end{equation*}
Now consider the quantity $(\bm{c}_0, T(\rho_{AB^{(k)}}) \bm{b}_0)$. Using \eqref{eq:tmeasurements_update} we may write this as 
\begin{equation}
\begin{aligned}
(\bm{c}_0, T(\rho_{AB^{(k)}}) \bm{b}_0) &= \frac{2 + \sqrt{1-\gamma_{k-1}^2}}{4} (\bm{c}_0, T(\rho_{AB^{(k-1)}}) \bm{b}_0) + \frac{\tr{(\id \otimes \sigma_{\bm{b}_0}) \rho_{AB^{(k-1)}} (\id \otimes \sigma_{\bm{b}_0})(\sigma_{\bm{c}_0} \otimes \sigma_{\bm{b}_0})}}{4} \\
&+(1-\sqrt{1-\gamma_{k-1}^2})\frac{\tr{(\id \otimes \sigma_{\bm{b}_1}) \rho_{AB^{(k-1)}} (\id \otimes \sigma_{\bm{b}_1})(\sigma_{\bm{c}_0} \otimes \sigma_{\bm{b}_0})}}{4}.
\end{aligned}
\end{equation}
Using the cyclicity of the trace this may be written more succinctly as 
\begin{equation}
\begin{aligned}
(\bm{c}_0, T(\rho_{AB^{(k)}}) \bm{b}_0) &= \frac{2 + \sqrt{1-\gamma_{k-1}^2}}{4} (\bm{c}_0, T(\rho_{AB^{(k-1)}}) \bm{b}_0) + \frac{\tr{ \rho_{AB^{(k-1)}}(\bm{c}_{0} \otimes \sigma_{\bm{b}_0}\sigma_{\bm{b}_0}\sigma_{\bm{b}_0})}}{4} \\
&+(1-\sqrt{1-\gamma_{k-1}^2})\frac{\tr{ \rho_{AB^{(k-1)}}(\bm{c}_{0} \otimes \sigma_{\bm{b}_1}\sigma_{\bm{b}_0}\sigma_{\bm{b}_1})}}{4}.
\end{aligned}
\end{equation}
We then use the identity $\sigma_{\bm{a}} \sigma_{\bm{b}} = (\bm{a},\bm{b})\id + \mathrm{i} \sigma_{\bm{a} \times \bm{b}}$ to simplify the matrix products in the above expressions. In particular we have
$\sigma_{\bm{b}_0}\sigma_{\bm{b}_0}\sigma_{\bm{b}_0}=\sigma_{\bm{b}_0}$ and
$\sigma_{\bm{b}_1}\sigma_{\bm{b}_0}\sigma_{\bm{b}_1}=-\sigma_{\bm{b}_0}$.
Inserting this into the above expression we get
\begin{equation}
\begin{aligned}
(\bm{c}_0, T(\rho_{AB^{(k)}}) \bm{b}_0) &= \left(\frac{2 + \sqrt{1-\gamma_{k-1}^2}}{4} + \frac14 - \frac{1-\sqrt{1-\gamma_{k-1}^2}}{4} \right)(\bm{c}_0, T(\rho_{AB^{(k-1)}}) \bm{b}_0) \\
&= \left(\frac{1 + \sqrt{1-\gamma_{k-1}^2}}{2} \right)(\bm{c}_0, T(\rho_{AB^{(k-1)}}) \bm{b}_0). 
\end{aligned}
\end{equation}
Performing an analogous calculation for $(\bm{c}_1, T(\rho_{AB^{(k)}}) \bm{b}_1)$ yields
\begin{equation}
\begin{aligned}
(\bm{c}_1, T(\rho_{AB^{(k)}}) \bm{b}_1) &= \left(\frac{2 + \sqrt{1-\gamma_{k-1}^2}}{4} - \frac14 + \frac{1-\sqrt{1-\gamma_{k-1}^2}}{4} \right)(\bm{c}_1, T(\rho_{AB^{(k-1)}}) \bm{b}_1) \\
&= \frac12(\bm{c}_1, T(\rho_{AB^{(k-1)}}) \bm{b}_1). 
\end{aligned}
\end{equation}
By recursion these give
\begin{align}
(\bm{c}_0, T(\rho_{AB^{(k)}}) \bm{b}_0) &= 2^{1-k}(\bm{c}_0, T(\rho_{AB^{(1)}}) \bm{b}_0)\prod_{j=1}^{k-1}(1+\sqrt{1-\gamma_j^2})\, \\
(\bm{c}_1, T(\rho_{AB^{(k)}}) \bm{b}_1) &= 2^{1-k}(\bm{c}_1, T(\rho_{AB^{(1)}}) \bm{b}_1)\,.
\end{align}
Inserting these into the CHSH expression \eqref{eq:chsh_expression2_app} and noting that $\sqrt{\lambda_i} = (\bm{c}_i, T(\rho_{AB^{(1)}}) \bm{b}_i)$ we obtain
\begin{equation}
\chshexp{k} = 2^{2-k}\left(\gamma_k \sqrt{\lambda_1} \sin(\theta) + \sqrt{\lambda_0}\cos(\theta)\prod_{j=1}^{k-1} \left(1+ \sqrt{1-\gamma_j^2}\right)\right).
\end{equation}
Note that for the maximally entangled state we have $\lambda_0 = \lambda_1 = 1$ and so in this case $\chshexp{k}$ reduces to the CHSH expression\ifarxiv~\eqref{eq:chsh}\else from the main text\fi.

\ifarxiv\section{Proof of Lemma~\ref{lem:gammak_increasing}}\else
\section{Proof of Lemma~1}\fi\label{app:lem1}
We prove this for the more general sequence (cf.\ the expression given in \ifarxiv Section~\ref{sec:unbound} of \fi the main text),
\begin{equation}\label{eq:gammak_app}
\gamma_k(\theta) := \begin{cases}
(1+\epsilon)\frac{2^{k-1} - \cos(\theta) \prod_{j=1}^{k-1} \left(1 + \sqrt{1-\gamma_j^2(\theta)}\right)}{\sqrt{\lambda_1}\sin(\theta)} \qquad &\text{if } 0 < \gamma_{k-1}(\theta) < 1 \\
\infty  &\text{otherwise }
\end{cases},
\end{equation}
with $\epsilon>0$, $\lambda_1>0$ and $\gamma_1(\theta) = (1+\epsilon)\frac{1-\cos(\theta)}{\sqrt{\lambda_1}\sin(\theta)}$.  The lemma as stated in the main text corresponds to the case where $\lambda_1 = 1$.

First we note that $\gamma_1(\theta) > 0$ for all $\theta \in (0,\pi/4]$ and $\epsilon>0$. Furthermore, $\gamma_1(\theta)<1$ whenever $0<\theta < \operatorname{arctan}\left(\tfrac{2 (1+\epsilon) \sqrt{\lambda_1}}{(1+\epsilon)^2 - \lambda_1}\right)$, so the sequence may admit additional finite terms.

If a term in the sequence is infinite then so are all subsequent terms. Suppose then that the $k^{\text{th}}$ term of the sequence is finite. This implies that $0 < \gamma_{j}(\theta) < 1$ for all $j = 1, \dots, k-1$ and in particular we have $1 < 1+\sqrt{1- \gamma_{k-1}^2(\theta)} < 2$. The bound $1+\sqrt{1- \gamma_{k-1}^2(\theta)} < 2$ then implies that $\gamma_k(\theta)/\gamma_{k-1}(\theta) > 2$ and so the finite terms of the sequence are strictly increasing and the sequence as a whole is a positive, increasing sequence.

\ifarxiv\section{Proof of Lemma~\ref{lem:gammak_limit}}\else
\section{Proof of Lemma~2}\fi\label{app:lem2}
Again, we prove this lemma for the more general sequence (cf.\ the expression given in \ifarxiv Section~\ref{sec:unbound} of \fi the main text),
\begin{equation}
\gamma_k(\theta) := \begin{cases}
(1+\epsilon)\frac{2^{k-1} - \cos(\theta) \prod_{j=1}^{k-1} \left(1 + \sqrt{1-\gamma_j^2(\theta)}\right)}{\sqrt{\lambda_1}\sin(\theta)} \qquad &\text{if } 0 < \gamma_{k-1}(\theta) < 1 \\
\infty  &\text{otherwise }
\end{cases},
\end{equation}
with $\epsilon>0$, $\lambda_1>0$ and $\gamma_1(\theta) = (1+\epsilon)\frac{1-\cos(\theta)}{\sqrt{\lambda_1}\sin(\theta)}$.  The lemma as stated in the main text corresponds to the case where $\lambda_1 = 1$.

We use the inequalities $\sqrt{1-x^{2}} \geq 1 - x^{2}$, for $x \in [0,1]$; $\cos(\theta) \geq 1-\theta^{2}/2$, for $\theta \in (0, \pi/4]$; and $\sin(\theta)\geq\theta/2$, for $\theta \in (0, \pi/4]$. Applying these to $\gamma_k(\theta)$ we have that when $\gamma_k(\theta)$ is finite,
\begin{equation}
\begin{aligned}
\gamma_k(\theta) &\leq (1+\epsilon) \frac{2^{k-1} - (1-\theta^{2}/2)\prod_{j=1}^{k-1}(2 - \gamma_j^2(\theta))}{\sqrt{\lambda_1}\theta/2} \\
&= (1+\epsilon)2^k \frac{1 - (1-\theta^{2}/2)\prod_{j=1}^{k-1}(1 - \gamma_j^2(\theta)/2)}{\sqrt{\lambda_1}\theta}.
\end{aligned}
\end{equation}
Now define a new sequence $p_k(\theta)$ based on this upper bound where
\begin{equation}\label{eq:pk_definition}
p_k(\theta) = \begin{cases}
(1+\epsilon)2^k \frac{1 - (1-\theta^{2}/2)\prod_{j=1}^{k-1}(1 - p_j^2(\theta)/2)}{\sqrt{\lambda_1}\theta} \qquad &\text{if } 0 < p_{k-1}(\theta) < 1 \\
\infty  &\text{otherwise }
\end{cases},
\end{equation}
and $p_1(\theta) = (1+\epsilon) \frac{\theta}{\sqrt{\lambda_1}}$. Since the right hand side of~\eqref{eq:pk_definition} is increasing with $\gamma_j(\theta)$ for any $j<k$, we have that $p_k(\theta) \geq \gamma_k(\theta)$. Therefore, $\gamma_k(\theta)$ is finite whenever $p_k(\theta)$ is finite. 


We now show that there always exists a $\theta_k \in (0, \pi/4]$ such that $0 < p_k(\theta) < 1$ for all $\theta \in (0, \theta_k)$ and that $\lim_{\theta \to 0^+} p_k(\theta) = 0$.

Firstly, note that $0 < p_1(\theta) < 1$ for all $\theta \in (0, \tfrac{\sqrt{\lambda_1}}{1+\epsilon})$ and that $p_1(\theta)$ is an odd polynomial in $\theta$. Therefore on the interval $(0, \tfrac{\sqrt{\lambda_1}}{1+\epsilon})$, $p_2(\theta)$ is an odd polynomial. As $p_2(\theta) > p_1(\theta) > 0$ and $p_2(\theta)$ has no constant term we know that $\lim_{\theta \to 0^+} p_2(\theta) = 0$ and therefore, as $p_2(\theta)$ is a continuous positive function on the interval, there exists some $\theta_2 \in (0, \tfrac{\sqrt{\lambda_1}}{1+\epsilon})$ such that $0 < p_2(\theta) < 1$ for all $\theta \in (0, \theta_2)$. We now proceed by induction to show that there exists some $\theta_k>0$ such that $p_k(\theta)$ is an odd polynomial on the interval $(0, \theta_k)$ and $0 < p_k(\theta) < 1$. 

Suppose there exists a $\theta_{k-1}$ such that on the interval $(0, \theta_{k-1})$ all $p_j(\theta)$ for $j = 1,\dots, k-1$ are odd polynomials in $\theta$ and $0 < p_j(\theta) < 1$. If this is the case then on the interval $(0, \theta_{k-1})$ $p_k(\theta)$ is finite and the numerator in~\eqref{eq:pk_definition} is an even polynomial in $\theta$ with no constant term. Cancelling the $\theta$ in the denominator, $p_k(\theta)$ is an odd polynomial in $\theta$. In particular, this implies that $\lim_{\theta \to 0^+} p_k(\theta) = 0$ and also, because $p_k(\theta) > 0$, there exists some $\theta_k \in (0, \theta_{k-1})$ such that $0 < p_k(\theta) < 1$ for all $\theta \in (0, \theta_k)$. It follows by induction that for any $n \in \bbn$ we can find a $\theta_n\in(0,\pi/4]$ such that $ 0 < p_1(\theta) < \dots < p_n(\theta) < 1$ for all $\theta \in (0, \theta_n)$.

Finally, as for each $n \in \bbn$ we have $0 <\gamma_n(\theta) \leq p_n(\theta)$ it follows that there exists a $\theta_n \in (0, \pi/4]$ such that $ 0 < \gamma_1(\theta) < \dots < \gamma_n(\theta) < 1$ for all $\theta \in (0, \theta_n)$ and that $\lim_{\theta \to 0^+} \gamma_n(\theta) = 0$.

\section{Unbounded violations for a larger set of two-qubit states}
Following the same proof strategy as presented in the main text, in order to observe $\chshexp{k} > 2$ for the more general measurement strategy \eqref{eq:general_strategy} we require
\begin{equation}
\gamma_k > \frac{2^{k-1} - \sqrt{\lambda_0}\cos(\theta) \prod_{j=1}^{k-1} \left(1 + \sqrt{1-\gamma_j^2}\right)}{\sqrt{\lambda_1}\sin(\theta)}.
\end{equation}
We can define a sequence $(\gamma_i(\theta))_i$ analogous to that \ifarxiv of~\eqref{eq:gammak_alt} \fi from the main text, i.e., for some fixed $\epsilon > 0$ we have
\begin{equation}
  \gamma_k(\theta) = \begin{cases} (1+\epsilon) \frac{2^{k-1} - \sqrt{\lambda_0}\cos(\theta) \prod_{j=1}^{k-1} \left(1 + \sqrt{1-\gamma_j^2}\right)}{\sqrt{\lambda_1}\sin(\theta)}&\text{if }\gamma_j\in (0,1)\text{ for all}j\in\{1,\dots,k-1\}\\
    \infty&\text{otherwise}\end{cases}.
\end{equation}
By \ifarxiv Lemma~\ref{lem:gammak_increasing} \else Lemma~1 \fi we have that for any non-zero $\lambda_0, \lambda_1$, $\gamma_k(\theta)$ is a strictly increasing sequence when finite. However, in the more general setting \ifarxiv Lemma~\ref{lem:gammak_limit} \else Lemma~2 \fi holds only when $\lambda_0 = 1$. From this, the proof of unbounded violations (\ifarxiv Theorem~\ref{thm:main} \else Theorem~1 \fi from the main text) can be replicated for the measurement strategy~\eqref{eq:general_strategy} when $\lambda_0 = 1$ and $\lambda_1 > 0$. We therefore arrive at the following, more general, theorem. 

\begin{theorem}\label{thm:main_generalised}
	Let $\rho_{AB^{(1)}}$ be an entangled two-qubit state and let $\lambda_0, \lambda_1$ be the two largest eigenvalues of the matrix $T(\rho_{AB^{(1)}})T^{T}(\rho_{AB^{(1)}})$. If $\lambda_0 = 1$ and $\lambda_1 > 0$, then for any $n \in \bbn$, there exists a sequence $(\gamma_i)_{i}$ and a $\theta \in (0, \pi/4]$ such that the measurement strategy~\eqref{eq:general_strategy} achieves
	\begin{equation}
	\chshexp{k} > 2,
	\end{equation}
	for all $k = 1,\dots,n$. 
\end{theorem}

As an example, for two-qubit states of the form
\begin{equation}\label{eq:explicit-family}
\begin{pmatrix}
\alpha & 0 & 0 & \beta \\
0 & 0 & 0 & 0 \\
0 & 0 & 0 & 0 \\
\beta^* & 0 & 0 & 1-\alpha
\end{pmatrix}
\end{equation}
where $\alpha\in[0,1]$ and $|\beta|\leq\sqrt{\alpha(1-\alpha)}$, we have $\lambda_0=1$ and $\lambda_1=4|\beta|^2$. Moreover, we have $\sigma_{\bm{b}_0}=\sigma_3$ and $\sigma_{\bm{b}_1}=\sigma_1$ and so their measurement strategy coincides with the strategy presented in the main text.\footnote{For this family of states the second and third largest eigenvalues coincide, i.e., $\lambda_1 = \lambda_2$, and so one could take $\sigma_{\bm{b}_1} = \sigma_2$ instead.} By the Schmidt-decomposition we can always write a pure two-qubit entangled state as $\ket{\psi} = \cos(\varphi) \ket{00} + \sin(\varphi) \ket{11}$ for $\varphi \in (0, \pi/4]$. As such, all pure two-qubit entangled states may be written in the form \eqref{eq:explicit-family} for some suitably chosen basis. Furthermore, we have $\lambda_1 = \sin^2(2 \varphi)$ which means that $\lambda_1 > 0$ for all $\varphi\in(0, \pi/4]$.

\begin{corollary}
	By initially sharing any pure, entangled two-qubit state, the number of Bobs that can violate the CHSH inequality with a single Alice is unbounded.
\end{corollary}

\begin{remark}
	There are also some mixed states that satisfy the conditions $\lambda_0 = 1$ and $\lambda_1 >0$. If the condition $\lambda_0 = 1$ holds then there exists some vector $\bm{c} \in \bbr^3$ of unit length such that $\tr{\rho(\sigma_{\bm c} \otimes \sigma_{\bm c})} = \pm 1$. If $\outer{c_+}$ and $\outer{c_-}$ are the projectors onto the eigenstates of $\sigma_{\bm c}$, then any mixture of $\rho$ with $\outer{c_+}\otimes\outer{c_+}$ or $\outer{c_-}\otimes\outer{c_-}$ will produce a state that still satisfies $\lambda_0 = 1$.
\end{remark}

\section{Bounds on the sequence $(\theta_n)$}\label{app:theta}
Suppose $\theta$ is chosen small enough such that the $k^{\text{th}}$ term in the sequence $(\gamma_k)_{k \in \bbn}$ is finite, i.e.
\begin{equation}\label{eq:gammak_remark}
\gamma_k(\theta) = (1+\epsilon)\frac{2^{k-1} - \cos(\theta)\prod_{j=1}^{k-1}\left(1+ \sqrt{1-\gamma_j^2(\theta)}\right)}{\sin(\theta)}.
\end{equation}
Using the inequalities $\sqrt{1-x^2} \geq 1 - x^2$ for $x \in [0,1]$, $\cos(\theta) \geq 1 - \theta^2/2$ for $\theta \in \bbr$, and $\sin(\theta) \geq \theta/2$ for $\theta \in [0,\pi/2]$ we arrive at the upper bound
\begin{equation}
\gamma_k(\theta) \leq (1+\epsilon) 2^{k} \frac{1- (1- \theta^2/2)\prod_{j=1}^{k-1} \left( 1 - \gamma_j^2(\theta)/2\right) } {\theta}.
\end{equation}
Similarly, using the inequalities $\sqrt{1-x^2} \leq 1 - x^2/2$ for $x \in [0,1]$, $\cos(\theta) \leq 1 - \theta^2/4$ for $\theta \in [0,\pi/4]$, and $\sin(\theta) \leq \theta$ for $\theta \in [0,\pi/2]$ we find the lower bound
\begin{equation}
\gamma_k(\theta) \geq (1+\epsilon) 2^{k-1} \frac{1- (1-\theta^2/4)\prod_{j=1}^{k-1} \left( 1 - \gamma_j^2(\theta)/4 \right) } {\theta}.
\end{equation}
We now define two sequences $p_k^{\uparrow}(\theta)$ and $p_k^{\downarrow}(\theta)$ based on the upper and lower bounds respectively:
\begin{equation}\label{eq:sk_up}
p^{\uparrow}_k(\theta) := (1+\epsilon) 2^{k} \frac{1- (1- \theta^2/2)\prod_{j=1}^{k-1} \left( 1 - p^{\uparrow}_j(\theta)^2/2\right) } {\theta},
\end{equation}
\begin{equation}\label{eq:sk_down}
p^{\downarrow}_k(\theta) := (1+\epsilon) 2^{k-1} \frac{1- (1-\theta^2/4)\prod_{j=1}^{k-1} \left( 1 - p^{\downarrow}_j(\theta)^2/4 \right) } {\theta},
\end{equation}
with $p^{\uparrow}_1(\theta) = (1+\epsilon) \, \theta$ and $p^{\downarrow}_1(\theta) = (1+\epsilon)\theta/4$. This gives us two sequences that satisfy $p^{\downarrow}_k(\theta) \leq \gamma_k(\theta) \leq p_k^{\uparrow}(\theta)$ whenever $p_j^{\uparrow}(\theta) < 1$ for all $j < k$. For all $k$, $p_k^{\uparrow}(\theta)$ and $p_k^{\downarrow}(\theta)$ are polynomials in $\theta$ with no constant term.  We proceed to estimate the growth of $\gamma_k(\theta)$ by bounding the growth of the coefficients of $\theta$. Let $c^{\uparrow}_k$ and $c^{\downarrow}_k$ be the coefficient of $\theta$ in $p^{\uparrow}_k(\theta)$ and $p_k^{\downarrow}(\theta)$ respectively. Then,
\begin{equation}
c^{\uparrow}_k = (1+\epsilon)2^{k-1}\left(1+\sum_{j=1}^{k-1} (c^{\uparrow}_j)^2 \right)
\end{equation}
and
\begin{equation}
c^{\downarrow}_k = (1+\epsilon)2^{k-3}\left(1 + \sum_{j=1}^{k-1} (c^{\downarrow}_j)^2 \right)
\end{equation}
with $c^{\uparrow}_1 = (1+\epsilon)$ and $c^{\downarrow}_1 = (1+\epsilon)/4$. Unfortunately, no general method exists for solving nonlinear recurrence relations of higher orders. A simpler lower bounding sequence for $c^{\downarrow}_k$ can be found by noting $\sum_{j=1}^{k-1} c_j^2 \geq c_{k-1}^2$. Using this we get a sequence
\begin{equation}\label{eq:dn_down}
\ddown_k =(1+\epsilon)2^{k-3} (\ddown_{k-1})^2\,.
\end{equation}
By direct computation we can compute $c^{\downarrow}_4$, a polynomial in $\epsilon$ with positive co-efficients, and with constant term 
larger than $4$.  We choose to start the sequence $\ddown_k$ at $\ddown_4=4\leq c^{\downarrow}_4$. Note that in principle we could start the sequence at $\ddown_1=(1+\epsilon)/4$, but starting at the fourth term, which is always larger than $1$, gives a tighter bound. (In addition, the choice of the bound of $4$ enables a neater final expression).
Now consider an upper bound for $c_k^{\uparrow}$, using $\sum_{j=1}^{k-1} c_j^2 \leq (k-1) c_{k-1}^2$ we have 
\begin{align*}
c_k^{\uparrow} &= (1+\epsilon)2^{k-1}\left(1 + \sum_{j=1}^{k-1}  (c^{\uparrow}_j)^2 \right)\\
&\leq (1+\epsilon)2^{k-1} k (c^{\uparrow}_{k-1})^2  \\
&\leq (1+\epsilon) 2^{2k - 1}  (c^{\uparrow}_{k-1})^2,
\end{align*}
where on the second line we used $k \leq 2^k$. Using this upper bound we define the new sequence 
\begin{equation}\label{eq:dn_up}
\dup_k =(1+\epsilon)2^{2k-1} (\dup_{k-1})^2,
\end{equation}
with $\dup_1 = (1+\epsilon)$.
The following two lemmas give us a closed form for $\dup_k$ and $\ddown_k$.
\begin{lemma}\label{lem:bk}
	Let $k \in \bbn$ with $k \geq k_0$ and $t, c \in \bbr$. Then the sequence
	\begin{equation}
	b_k := 2 b_{k-1} + tk + c
	\end{equation}
	admits the closed form solution
	\begin{equation}
          b_k=2^{k-1}(b_{k_0}+(2+k_0)t+c)-(k+2)t-c\,.
         \end{equation}
         \begin{proof}
		Rearranging, we find the relation 
		$$
		b_k - 2 b_{k-1} = t k + c
		$$
		and through a telescoping procedure we can write 
		\begin{align*}
		b_k-2^{k-k_0}b_{k_0} &= \sum_{r=0}^{k-k_0-1} 2^r(b_{k-r}-2b_{k-r-1}) \\
		&= \sum_{r=0}^{k-k_0-1} 2^r (t(k-r)+c).
		\end{align*}
		Evaluating the summation we find
		\begin{align*}
		b_k-2^{k-k_0}b_{k_0} &=2^{k-k_0}((2+k_0)t+c)-(k+2)t-c
		\end{align*}
		and so 
		\begin{align*}
		b_k &= 2^{k-k_0}(b_{k_0}+(2+k_0)t+c)-(k+2)t-c.\qedhere
		\end{align*}
	\end{proof}
\end{lemma}
\begin{lemma}
	Let $\epsilon>0$ and $k \in \bbn$ with $k \geq 1$. Let 
	\begin{align*}
		\ddown_k &=  (1+\epsilon)2^{k-3} (\ddown_{k-1})^2
		\intertext{and}
		\dup_k &= (1+\epsilon)2^{2k-1} (\dup_{k-1})^2
	\end{align*}
	define two sequences with $\ddown_4=4$ and $\dup_1 = (1+\epsilon)$. Then for $k\geq4$ we have
	\begin{align*}
		\ddown_k &= (1+\epsilon)^{2^{k-4}-1}\,2^{5\times2^{k-4}-k+1}
		\intertext{and for $k\geq1$ we have}
		\dup_k &= (1+\epsilon)^{2^{k}-1}\,2^{5\times 2^{k-1} - 2 k - 3}\,.
	\end{align*}  
	\begin{proof}
		We derive a closed-form expression for the sequence
		$$
		d_k = (1+\epsilon)2^{tk-a} d_{k-1}^{2},
		$$ 
		for $t,a,k \in \bbn$ and $k\geq k_0$. Taking logarithms of both sides we have 
		$$
		\log(d_k) = (tk-a) + \log(1+\epsilon) + 2 \log(d_{k-1}).
		$$
		Defining a new sequence via the replacements $b_k = \log(d_k)$ and $c = \log(1+\epsilon) - a$, we may apply Lemma~\ref{lem:bk} to get
		\begin{align*}
		b_k&=2^{k-k_0}(b_{k_0}+(2+k_0)t+c)-(k+2)t-c.
		\end{align*}
	Equivalently,
	\begin{align*}
	d_k&=2^{2^{k-k_0}(\log(d_{k_0})+(2+k_0)t+c)-(k+2)t-c} \\
	&=2^{2^{k-k_0}(\log(d_{k_0})+(2+k_0)t+\log(1+\epsilon)-a)-(k+2)t-\log(1+\epsilon)+a}
	\end{align*}
	To recover $\ddown_k$ we first set $t=1$, $a=3$, $k_0=4$ and $d_4=4$ to obtain
	\begin{align*}
          \ddown_k&=2^{2^{k-4}(5+\log(1+\epsilon))-(k+2)-\log(1+\epsilon)+3}\\
	&= 2^{(2^{k-4}-1)(\log(1+\epsilon)}2^{5\times2^{k-4}-k+1}\\
	&= (1+\epsilon)^{2^{k-4}-1}\,2^{5\times2^{k-4}+1-k}.
	\end{align*}
	Setting $t=2$, $a=1$, $k_0=1$ and $d_1=(1+\epsilon)$ we find
	\begin{align*}
		\dup_k &= 2^{2^{k-1}(\log(1+\epsilon)+\log(1+\epsilon) + 5) - 2(k+2) - \log(1+\epsilon) + 1}  \\
		&= 2^{2^{k-1}(2\log(1+\epsilon)+5)-(2k+3)-\log(1+\epsilon)} \\
		&= 2^{(2^k-1)\log(1+\epsilon)+5\times 2^{k-1}-2k-3} \\
		&= (1+\epsilon)^{2^k-1}\,2^{5\times 2^{k-1}-2k-3}\qedhere
	\end{align*}
	\end{proof}
\end{lemma}
Hence both $\dup_k$ and $\ddown_k$ exhibit double-exponential growth with $k$.  Relating this back to $\gamma_n(\theta)$, from the proof \ifarxiv of Lemma~\ref{lem:gammak_limit} \else in Section~\ref{app:lem2} \fi we know that there exists some $\theta_n \in (0, \pi/4)$ such that on the interval $(0, \theta_n)$ we have $0 < p_k^{\downarrow}(\theta) \leq \gamma_k(\theta) \leq p_k^{\uparrow}(\theta)$ for all $k < n$. Furthermore, from the two previous lemmas we have established that there exist polynomials $d_n^{\downarrow} \theta + O(\theta^3)$ and $d_n^{\uparrow} \theta + O(\theta^3)$ such that
\begin{equation}
d_n^{\downarrow}\theta+O(\theta^3)\leq\gamma_n(\theta)\leq d_n^{\uparrow}\theta+O(\theta^3)
\end{equation}
on this interval. The two inequalities together imply that $\gamma_n(\theta)$ grows double-exponentially fast in $n$ for sufficiently small $\theta$, which suggests that we may need $\theta$ to be double exponentially small to observe $\gamma_n(\theta) \leq 1$. However, because we have not ruled out that the higher order terms remain significant when $\theta_n$ is double-exponentially small in $n$, our argument is not conclusive, although a double-exponental scaling agrees with the numerics presented in \ifarxiv Fig.~\ref{fig:numerics}\else the main text\fi.

\end{document}